\documentclass[a4paper,UKenglish]{lipics}

\usepackage{amsbsy,eucal}
\usepackage{amsmath}
\usepackage{amssymb}
\usepackage{enumerate}
\usepackage{proof}
\usepackage{wrapfig}

\usepackage{pifont}
\usepackage{pxfonts} 
\usepackage{tikz}
\usetikzlibrary{arrows,shapes,snakes,automata,backgrounds,petri,calc}
\colorlet{dblue}{blue!40!black}

\theoremstyle{plain}
\newtheorem{proposition}[theorem]{Proposition}
\theoremstyle{definition}

\newcommand{\shorten}[1]{}

 \renewcommand{\hbar}{\overline{h}}

\newcommand{\set}[1]{\{ #1 \}}

\renewcommand{\phi}{\varphi}

\newcommand{\nat}{\mathbb{N}}

\newcommand{\emptyword}{\varepsilon}

\newcommand{\thuemorse}{\mbox{Thue\hspace{.08em}--Morse}}

\newcommand{\pairlft}{(}
\newcommand{\pairrgt}{)}
\newcommand{\pairsep}{{,\,}}
\newcommand{\pairstr}[1]{\pairlft#1\pairrgt}
\newcommand{\pair}[2]{\pairstr{#1\pairsep#2}}
\newcommand{\triple}[2]{\pair{#1\pairsep#2}}

\newcommand{\quadruple}[2]{\triple{#1\pairsep#2}}
\newcommand{\quintuple}[2]{\quadruple{#1\pairsep#2}}
\newcommand{\sixtuple}[2]{\quintuple{#1\pairsep#2}}

\newcommand{\funap}[2]{#1(#2)}
\newcommand{\bfunap}[3]{#1(#2,#3)}

\newcommand{\shastyp}{{:}}
\newcommand{\hastyp}[2]{#1 \mathrel{\shastyp} #2}

\newcommand{\sub}[2]{#1_{#2}}

\newlength{\mathfrwidth}
\newsavebox{\mathfrbox}

\newcommand{\rem}[1]{\relax}

\newcommand{\nb}{\nobreakdash}

\newcommand{\blocks}[1]{\mathsf{blocks}(#1)}
\newcommand{\unblock}[1]{\mathsf{cat}(#1)}
\newcommand{\img}[1]{\mathsf{img}(#1)}

\newcommand{\erase}[2]{\gamma_{#1}(#2)}
\newcommand{\tail}[1]{\mathsf{tail}(#1)}
\newcommand{\head}[1]{\mathsf{head}(#1)}

\newcommand{\occ}[2]{|#2|_{#1}} 
\newcommand{\smatrix}[1]{M_{#1}} 

\newcommand{\floor}[1]{\lfloor #1 \rfloor}

\newcommand{\dead}{\mathcal{D}}

\newcommand{\aalph}{\Sigma}
\newcommand{\balph}{\Gamma}

\newcommand{\afst}{M} 
\newcommand{\firstfst}{\afst} 
\newcommand{\secondfst}{N} 
\newcommand{\alphin}{\Sigma}
\newcommand{\alphout}{\Delta}

\newcommand{\states}{Q} 
\newcommand{\astate}{q} 
\newcommand{\istate}{\sub{\astate}} 
\newcommand{\stransfun}{\delta} 
\newcommand{\transfun}{\bfunap{\stransfun}} 
\newcommand{\transfunz}{\funap{\stransfun}} 
\newcommand{\soutfun}{\lambda} 
\newcommand{\outfun}{\bfunap{\soutfun}} 
\newcommand{\outfunz}{\funap{\soutfun}} 
\newcommand{\outa}[2]{\funap{#1}{#2}} 

\title{Eigenvalues and Transduction of Morphic Sequences: Extended Version\footnote{
  This is an extended version of our paper~\cite{dlt} presented at \emph{Developments in Language Theory 2014}.
  This extended version contains examples and additional remarks.
}}

\author[1]{David Sprunger}
\author[2]{William Tune}
\author[3]{J\"{o}rg Endrullis}
\author[4]{Lawrence S.~Moss}

\affil[1]{Department of Mathematics, Indiana University, Bloomington IU 47405 USA.}
\affil[2]{Department of Mathematics, Indiana University, Bloomington IU 47405 USA.}
\affil[3]{Vrije Universiteit Amsterdam, Department of Computer Science, 
    1081 HV Amsterdam, The Netherlands;  and
    Department of Mathematics, Indiana University, Bloomington IU 47405 USA.
}  
\affil[4]{
   Department of Mathematics, Indiana University, Bloomington IU 47405 USA.
   This work was partially supported by a grant from the Simons Foundation (\#245591 to Lawrence Moss).
}

\authorrunning{D. Sprunger, W. Tune, J. Endrullis and L. S.~Moss}

\Copyright{D. Sprunger, W. Tune, J. Endrullis and L. S.~Moss}

\serieslogo{}
\volumeinfo
 {}
 {1}
 {}
 {1}
 {1}
 {1}
\EventShortName{}

\def\copyrightline{}

\begin{document}

\maketitle

\begin{abstract}
  We study finite state transduction of automatic and morphic sequences.
  Dekking~\cite{dekk:94} proved that
  morphic sequences are closed under transduction and in particular morphic images.
  We present a simple proof of this fact,
  and use the construction in the proof to show that non\nb-erasing transductions preserve
  a condition called  $\alpha$-substitutivity. 
  Roughly, a sequence is $\alpha$-substitutive if the sequence can be obtained as 
  the limit of iterating a substitution with dominant eigenvalue $\alpha$.
 Our results culminate in the following fact:  for multiplicatively independent real numbers $\alpha$ and $\beta$,
  if $v$ is a $\alpha$-substitutive sequence and $w$ is an $\beta$-substitutive sequence, then $v$ and $w$ have
  no common non-erasing transducts except for the ultimately periodic sequences. 
  We rely on Cobham's theorem for substitutions, a recent result of Durand~\cite{dura:2011}.
\end{abstract}

\section{Introduction}

Infinite sequences of symbols are of paramount importance in a wide range of fields,
ranging from formal languages to pure mathematics and physics.
A landmark was the discovery in 1912 by Axel Thue, founding father of formal language theory,
of the famous sequence
$
  0110\; 1001\; 1001\; 0110\; 1001\; 0110\; \cdots 
$.
Thue was interested in infinite words which avoid certain patterns,
like squares $ww$ or cubes $www$, when $w$ is a non-empty word.
Indeed, the sequence shown above, called the  \thuemorse{} sequence, is cube-free.  
It is perhaps the most natural cube-free infinite word.

\begin{wrapfigure}{r}{5cm}
\vspace{-4ex}
\begin{center}
  \includegraphics{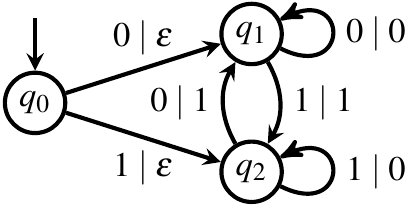}
  \vspace{-.2cm}
  \caption{A transducer computing the difference (exclusive or) of consecutive bits.}
  \label{fig:fst}
\end{center}
\vspace{-3ex}
\end{wrapfigure}
A common way to transform infinite sequences is by using \emph{finite state transducers}.
These transducers are deterministic finite automata with input letters and output words for each 
transition; an example is shown in Figure~\ref{fig:fst}.  Usually we omit the words ``finite state'' 
and refer to \emph{transducers}. A transducer maps infinite sequences to infinite sequences by 
reading the input sequence letter by letter. Each of these transitions produces an output word,
and the sequence formed by concatenating each of these output words in the order they 
were produced is the output sequence. In particular, since this transducer runs for infinite time 
to read its entire input, this model of transduction does not have final states. A transducer is called 
\emph{$k$-uniform} if each step produces $k$-letter words.
For example,  \emph{Mealy machines}  are  $1$-uniform transducers.
A transducer is \emph{non-erasing} if each step produces a non-empty word;
this condition is prominent in this paper.

Although transducers are a natural machine model,
 hardly anything is known about their capabilities of transforming infinite sequences.
To state the issues more clearly, let us write $x \trianglelefteq y$ if there is a transducer taking $y$ to $x$.
This transducibility gives rise to a partial order of \emph{stream degrees}~\cite{endr:hend:klop:2011}
that is analogous to, but more fine-grained than,  recursion-theoretic orderings such as
\emph{Turing reducibility} $\leq_T$ and \emph{many-one reducibility} $\leq_m$.
We find it surprising that so little is known about $\trianglelefteq$.     
As of now, the structure of this order is vastly unexplored territory with many open questions.
To answer these questions, we need  a better understanding of  transducers.

The main things that are known at this point concern two particularly well-known sets of streams, namely the \emph{morphic} and \emph{automatic} sequences.
Morphic sequences are obtained as the limit of iterating a morphism on a starting word 
(and perhaps applying a coding to the limit word).
Automatic sequences have a number of independent characterizations (see~\cite{allo:shal:2003}); we shall not
repeat these here. 
 There are two seminal closure results concerning the transduction of morphic and automatic sequences:
\begin{enumerate}
\item[(1)] The class of morphic sequences is closed under  transduction (Dekking~\cite{dekk:94}).
\item[(2)] For all $k$, the class of $k$-automatic sequences  is closed under uniform transduction (Cobham~\cite{cobh:72}).
\end{enumerate}

The restriction in (2) to uniform transducers is shown by the following example.
\begin{example}\label{ex:automatic:not:closed}
  Let $w \in \{\,0,1\,\}^\omega$ be defined by $w(n) = 1$ if $n$ is a power of $2$ and $w(n) = 0$ otherwise.
  This sequence is $2$-automatic.
  Let $h$ be the morphism $0 \mapsto 0$ and $1 \mapsto 01$.
  Taking the image of $w$ under $h$, that is $h(w)$, yields a sequence that is no 
  longer automatic (but still morphic).
  Here is a sketch that $h(w)$ is not $2$-automatic. 
  Note that the $i^{th}$ digit in $h(w)$ is $1$ iff $i = 2^n + n$ for
  some $n$.   Suppose that $M$ is a finite-state machine 
  with the property that reading in each number $i$ in binary yields the $i^{th}$ digit of $h(w)$.
  Let $N$ be large enough so that the binary representation of $2^N+N$ has a run of zeroes longer than the number of
  states in $N$.  Then by pumping, $N$ must accept a number which is not of the form $2^n+n$.
\end{example}

In this paper,
we do not attack the central problems concerning the stream degrees.
Instead, 
 we are interested in a closure result for non-erasing transductions.  Our interest comes from the following easy observation:
\begin{enumerate}
\item[(3)]\label{item:everything}
 For every morphic sequence $w \in \Sigma^\omega$ 
there is a $2$-automatic sequence $w' \in (\Sigma\cup \{\,a\,\})^\omega$ 
such that $w$ is obtained from $w'$ by erasing all occurrences of \shorten{the letter }$a$.
(See Allouche and Shallit~\cite[Theorem 7.7.1]{allo:shal:2003})
\end{enumerate}
\smallskip

\noindent
This motivates the question: how powerful is  non-erasing transduction?

\paragraph*{Our contribution}
The main result of this paper is stated in terms of the notion of \emph{$\alpha$-substitutivity}.
This condition is defined in Definition~\ref{definition-alpha-substitutive} below, and the definition uses the eigenvalues of matrices naturally associated with morphisms on finite alphabets.   Indeed, the core of our work is
a collection of results on eigenvalues of these matrices.

 We prove that  the set of $\alpha$-substitutive words is closed
under non\nb-erasing finite state transduction.
We follow Allouche and Shallit~\cite{allo:shal:2003}
in obtaining   transducts of a given morphic sequence $w$ by
 \emph{annotating} an iteration morphism, 
and then taking a morphic image of the annotated limit sequence.
For the first part of this transformation, 
we show that a morphism and its annotation have the same 
eigenvalues with non-negative eigenvectors.
For the second part, 
we revisit the proof given in Allouche and Shallit~\cite{allo:shal:2003}  of Dekking's theorem
\shorten{mentioned in (1) above }%
that morphic images of morphic sequences are morphic.
We simplify the construction in the proof to make it amenable
for an analysis of the eigenvalues of the resulting~morphism.
\shorten{As a  side-effect of the simplified construction, we shorten the proof of (1) significantly.}

\paragraph*{Related work}
Durand~\cite{dura:2011} proved that 
if $w$ is an $\alpha$-substitutive sequence and $h$ is a non-erasing morphism,
then $h(w)$ is $\alpha^k$-substitutive for some $k\in\nat$.
We strengthen this result in two directions. First, we show that  $k$ may be taken to be $1$;
 hence $h(w)$  is $\alpha^k$-substitutive for every $k\in\nat$.
Second, we show that Durand's result also holds for non-erasing transductions.


\section{Preliminaries}

We recall some of the main concepts that we use in the paper.
For a thorough introduction to morphic sequences, automatic sequences and finite state transducers,
we refer to~\cite{allo:shal:2003,saka:03}. 

We are concerned with infinite sequences $\aalph^\omega$ over a finite alphabet $\aalph$.
We write $\aalph^*$ for the set of finite words,
$\aalph^+$ for the finite, non-empty words, 
$\aalph^\omega$ for the infinite words, 
and $\aalph^\infty = \aalph^* \cup \aalph^\omega$ for all finite or infinite words over $\aalph$.

\subsection{Morphic sequences and automatic sequences}

\begin{definition}
  A \emph{morphism} is a map 
  $h : \aalph \to \balph^*$.
  This map extends by concatenation to 
  $h : \aalph^* \to \balph^*$, and we do not distinguish the two notationally.
  Notice also that $h(vu) = h(v)h(u)$ for all $u,v\in \aalph^*$.
  If $h_1, h_2 :\aalph \to \aalph^*$, we have a composition 
  $h_2\circ h_1: \aalph \to \aalph^*$.

  An \emph{erased letter} (with respect to $h$) is some $a\in \aalph$ such that $h(a) = \emptyword$.
  A morphism $h : \aalph^* \to \balph^*$ is called \emph{erasing} if has an erased letter.
  A morphism is \emph{$k$-uniform} (for $k\in\nat$) if $|h(a)| = k$ for all $a \in \aalph$.
  A \emph{coding} is a $1$-uniform morphism $c : \aalph \to \balph$. 
\end{definition}

A morphic sequence is obtained by iterating a morphism, and applying a coding to the limit word.
\begin{definition}
  Let $s \in \aalph^+$ be a word, $h : \aalph \to \aalph^*$ a morphism, and $c : \aalph \to \balph$ a coding. 
  If the limit     
  $
    h^\omega(s) = \lim_{n\to \infty} h^n(s)
  $ 
  exists and is infinite, then $h^\omega(s)$ is a \emph{pure morphic sequence},
  and $c(h^\omega(s))$ a \emph{morphic sequence}.
  
  If $h(x_1) = x_1 z$ for some  $z\in \Sigma^+$, 
  then we say that $h$ is \emph{prolongable} on $x_1$.  
  In this case,  $h^\omega(x_1)$ is a pure morphic sequence.

  If additionally, the morphism $h$ is $k$-uniform, then $c(h^\omega(s))$ is a \emph{$k$-automatic sequence}.
  A sequence $w \in \aalph^\omega$ is called \emph{automatic} if $w$ is $k$-automatic for some $k\in\nat$.
\end{definition}

\begin{example}
  A well-known example of a purely morphic word is the \thuemorse{} sequence.
  This sequence can be obtained as the limit of iterating the morphism
  $0 \mapsto 01$, $1 \mapsto 10$
  on the starting word~$0$.
  The first iterations are
  $$0 
    \mapsto\ 01 
    \mapsto\ 0110 
    \mapsto\ 01101001
    \mapsto\ 01101001 10010110
    \mapsto\ \cdots \;,
  $$
  and they converge, in the limit, to the \thuemorse{} sequence.
  As the morphism $h$ is $2$\nobreakdash-uniform, the sequence is also $2$-automatic.
\end{example}

\begin{example}\label{example-Fibonacci}
  An example of a purely morphic word which is not automatic is provided by the 
  \emph{Fibonacci substitution} $a\mapsto ab$, $b\mapsto a$.
  Starting with $a$, the fixed point is $$abaababaabaababaababaabaababaabaababaaba\cdots\;.$$
\end{example}

\subsection{Cobham's Theorem for morphic words}

\begin{definition}
  For $a \in \aalph$ and $w \in \aalph^*$ 
  we write $\occ{a}{w}$ for the \emph{number of occurrences of $a$ in $w$}. 
  Let $h$ be a morphism over $\aalph$. 
  The \emph{incidence matrix} of $h$ is the matrix $\smatrix{h} = (m_{i,j})_{i\in \aalph,j\in\aalph}$
  where $m_{i,j} = \occ{i}{h(j)}$ is the number of occurrences of the letter $i$ in the word $h(j)$.
\end{definition}

\begin{theorem}[Perron-Frobenius]\label{thm:perron}
  Every non-negative square matrix $M$ 
  has a real eigenvalue $\alpha \ge 0$ 
  that is greater than or equal to the absolute value of any other eigenvalue of $M$
  and the corresponding eigenvector is non-negative.
  We refer to $\alpha$ as the \emph{dominating eigenvalue} of $M$.
\end{theorem}

\begin{definition}\label{definition-alpha-substitutive}
  The \emph{dominating eigenvalue} of a morphism $h$ is the dominating eigenvalue of $\smatrix{h}$.
  An infinite sequence $w \in \aalph^\omega$ 
  over a finite alphabet $\aalph$
  is said to be \emph{$\alpha$-substitutive} ($\alpha \in \mathbb{R}$)
  if there exist
  a morphism $h : \aalph \to \aalph^*$ with dominating eigenvalue $\alpha$,
  a coding $c : \aalph \to \aalph$
  and a letter $a \in \aalph$
  such that
  (i)~$w = c(h^\omega(a))$,
  and (ii) every letter of $\aalph$ occurs in $h^\omega(a)$.
\end{definition}

\begin{remark}
  Let us remark on the importance of the condition (ii) in Definition~\ref{definition-alpha-substitutive}.
  Without this condition every $\alpha$-substitutive sequence $w \in \Sigma^\omega$ would
  also be $\beta$-substitutive for every $\beta > \alpha$ that is the dominating eigenvalue of a non-negative integer matrix.

  This can be seen as follows.
  Let $h : \Sigma \to \Sigma^*$ be a morphism with dominating eigenvalue $\alpha$.
  Let $a \in \Sigma$ such that $w = h^\omega(a)$ exists, is infinite and contains all letters from $\Sigma$.
  Then $w$ is $\alpha$-substitutive.
  Now let $\beta > \alpha$ be the dominating eigenvalue of a non-negative integer matrix.
  Then there exists an alphabet $\Gamma$ (disjoint from $\Sigma$, $\Gamma \cup \Sigma = \varnothing$) and 
  a morphism $g : \Gamma \to \Gamma^*$ with dominating eigenvalue $\beta$.
  Define $z : (\Sigma\cup \Gamma) \to (\Sigma\cup \Gamma)^*$
  by $z(b) = h(b)$ for all $b \in \Sigma$ and $z(c) = h(c)$ for all $c \in \Gamma$.
  Then $z^\omega(a) = h^\omega(a) = w$ and
  the dominating eigenvalue of $z$ is $\beta$.
\end{remark}

Two complex numbers $x,y$ are called \emph{multiplicatively independent}
if for all $k,\ell \in \mathbb{Z}$ it holds that $x^k = y^\ell$ implies $k = \ell = 0$.
We shall use the following version of Cobham's theorem due to Durand~\cite{dura:2011}.

\begin{theorem}\label{theorem-Durand}
  Let $\alpha$ and $\beta$ be multiplicatively independent Perron numbers.
  If a sequence $w$ is both $\alpha$-substitutive and $\beta$-substitutive,
  then $w$ is eventually periodic.
  \qed
\end{theorem}

\subsection{Transducers}

\begin{definition}
  A \emph{(sequential finite-state stream) transducer (FST)} 
  $\afst = \sixtuple{\alphin}{\alphout}{\states}{\istate{0}}{\stransfun}{\soutfun}$
  consists of
  \begin{enumerate}[(i)]
    \item a finite input alphabet $\alphin$,
    \item a finite output alphabet $\alphout$,
    \item a finite set of states $\states$,
    \item an initial state $\istate{0} \in \states$,
    \item a transition function $\hastyp{\stransfun}{\states\times \alphin \to \states}$, and
    \item an output function $\hastyp{\soutfun}{\states\times \alphin \to \alphout^*}$.
  \end{enumerate}
\end{definition}

\begin{example}
The  transducer $\sixtuple{\alphin}{\alphout}{\states}{\istate{0}}{\stransfun}{\soutfun}$ 
shown in Figure~\ref{fig:fst} 
can be defined as follows:
$\alphin = \alphout = \{\,0,1\,\}$,
$\states = \{\,\istate{0},\istate{1},\istate{2}\,\}$ with $q_0$ the initial state, 
and the transition function $\stransfun$ and output function $\soutfun$ are given by:
\begin{align*}
  \bfunap{\stransfun}{\istate{0}}{0} &= \istate{1}
  & \bfunap{\soutfun}{\istate{0}}{0} &= \emptyword 
  & \bfunap{\stransfun}{\istate{0}}{1} &= \istate{2}
  & \bfunap{\soutfun}{\istate{0}}{1} &= \emptyword \\
  \bfunap{\stransfun}{\istate{1}}{0} &= \istate{1}
  & \bfunap{\soutfun}{\istate{1}}{0} &= 0
  & \bfunap{\stransfun}{\istate{1}}{1} &= \istate{2}
  & \bfunap{\soutfun}{\istate{1}}{1} &= 1 \\
  \bfunap{\stransfun}{\istate{2}}{0} &= \istate{1}
  & \bfunap{\soutfun}{\istate{2}}{0} &= 1
  & \bfunap{\stransfun}{\istate{2}}{1} &= \istate{2}
  & \bfunap{\soutfun}{\istate{2}}{1} &= 0
\end{align*}
\end{example}

We use transducers to transform infinite words.
The transducer reads the input word letter by letter,
and the transformation result is the concatenation of the output words encountered along the edges.

\begin{definition}\normalfont
  Let 
  $\afst = \sixtuple{\alphin}{\alphout}{\states}{\istate{0}}{\stransfun}{\soutfun}$
  be a  transducer.
  We extend the state transition function $\stransfun$ 
  from letters $\alphin$ to finite words $\alphin^*$ as follows:
  $\transfun{\astate}{\emptyword} = \astate$ and
  $\transfun{\astate}{aw} = \transfun{\transfun{\astate}{a}}{w}$
  for $\astate \in \states$, $a \in \alphin$, $w \in \alphin^*$.
  
  The output function $\soutfun$ is extended to the set of all words
  $\alphin^\infty = \alphin^\omega \cup \alphin^*$
  by the following definition:
  $\outfun{\astate}{\emptyword} = \emptyword$ and
  $\outfun{\astate}{aw} = \outfun{\astate}{a} \, \outfun{\transfun{\astate}{a}}{w}$
  for $\astate \in \states$, $a \in \alphin$, $w \in \alphin^\infty$.

  We introduce $\transfunz{w}$ and $\outfunz{w}$
  as shorthand for $\transfun{\istate{0}}{w}$ and $\outfun{\istate{0}}{w}$, respectively.
  Moreover, we define $\outa{\afst}{w} = \outfunz{w}$,
  the \emph{output of $\afst$ on $w \in \alphin^\omega$}.    In this way, we think of $\afst$ as a function
  from (finite or infinite) words on its input alphabet to infinite words on its output alphabet
  $\afst: \alphin^\infty\to\alphout^\infty$.   
  
  If $x\in \alphin^\omega$ and $y\in \alphout^\omega$, we write $y \trianglelefteq x$ if for some 
  transducer $\afst$, we have $\afst(x) = y$.
\end{definition}

Notice that every morphism is computable by a transducer (with one state).  
In particular, every coding is computable by a transducer.
\begin{definition}\label{definition-composition-transducers}
  Let $\firstfst = \sixtuple{\alphin}{\alphout}{\states}{\istate{0}}{\stransfun}{\soutfun}$
  and $\secondfst = \sixtuple{\alphin'}{\alphout'}{\states'}{\istate{0}'}{\stransfun'}{\soutfun'}$
  be transducers, and assume that $\alphin' = \alphout$.
  We define the \emph{composition} $\secondfst\circ\firstfst$ 
  to be the transducer  
  \begin{align*}
  \secondfst\circ\firstfst 
  = \sixtuple{\;\;&\alphin}
    {\;\;\alphout'}
    {\;\;\states\times \states'}
    {\;\;\pair{\istate{0}}{\istate{0}'}}
    {\\
    &\;\;\pair{\pair{q}{q'}}{a} \mapsto \pair{\transfun{q}{a}}{\bfunap{\stransfun'}{q'}{\outfun{q}{a}}}}
    {\\
    &\;\;\pair{\pair{q}{q'}}{a} \mapsto \soutfun'(q',\soutfun(q,a))\;\;}\;.
  \end{align*}
  Here $\stransfun'$ and  $\soutfun'$ are the extensions of the transition and output functions of $\secondfst$ to $\alphin^*$, respectively.
\end{definition}
\begin{proposition}\label{prop-basic}
  Concerning the composition relation on transducers and $\trianglelefteq$ on finite and infinite words:
  \begin{enumerate}[(i)]
    \item  
      The map $\alphin^\infty\to (\alphout')^\infty$ computed by \shorten{the transducer}%
      $\secondfst\circ\firstfst$ is \shorten{exactly }the composition
      of $\firstfst: \alphin^\infty\to\alphout^\infty$ followed by 
      $\secondfst: \alphout^\infty\to(\alphout')^\infty$.
    \item The relation $\trianglelefteq$ is transitive.
    \item  If $x\in \alphin^\infty$ and $h: \alphin\to\alphout^*$ is a coding, then $h(x) \trianglelefteq x$.
  \end{enumerate}
\end{proposition}

\section{Closure of Morphic Sequences under Morphic Images}

\label{section-morphic-images}

\begin{definition}\normalfont
  Let $h : \Sigma^* \to \Sigma^*$ be  morphisms, and let $\Gamma \subseteq \Sigma$  be a set  of  letters.
  We call a letter $a\in\Sigma$
  \begin{enumerate}[(i)]
    \item \emph{dead} if $h^n(a) \in \Gamma^*$ for all $n \ge 0$,
    \item \emph{near dead} if $a\notin \Gamma$, and for all $n >0$,   $h^n(a)$ consists of dead letters,
    \item \emph{resilient} if $h^n(a) \not\in \Gamma^*$ for all $n \ge 0$,
    \item \emph{resurrecting} if $a\in\Gamma$ and $h^n(a) \not\in \Gamma^*$ for all $n > 0$.
  \end{enumerate} 
  \smallskip
  
  \noindent
  with respect to $h$ and $\Gamma$.
  We say that the morphism $h$ \emph{respects} $\Gamma$
  if every letter $a \in \Sigma$ is either dead, near dead, resilient, or resurrecting.
  (Note that all of these definitions are with respect to some fixed $h$ and $\Gamma$.)
\end{definition}

\begin{lemma}\label{lem:respect}
  Let $g: \Sigma^* \to \Sigma^*$ be a morphism, and let $\Gamma \subseteq \Sigma$\shorten{ be a set of letters}.
  Then $g^r$ respects $\Gamma$ for some natural number $r > 0$.
\end{lemma}
\begin{proof}
  See Lemma 7.7.3 in Allouche and Shallit~\cite{allo:shal:2003}.
\end{proof}

\begin{definition}
  For a set of letters $\Gamma \subseteq \Sigma$ and a word $w \in \Sigma^\infty$,
  we write $\erase{\Gamma}{w}$
  for the word obtained from $w$ by erasing all occurrences of letters in $\Gamma$.
\end{definition}

\begin{definition}\normalfont\label{def:close:erase}
  Let $g : \Sigma^*\to \Sigma^*$ be a morphism, and $\Gamma \subseteq \Sigma$ a set of letters.
  We construct an alphabet $\Delta$, a morphism $\xi : \Delta^*\to \Delta^*$ 
  and a coding $\rho : \Delta \to \Sigma$ as follows.
  We refer to $\Delta,\xi,\rho$ as the \emph{morphic system associated with the erasure of $\Gamma$ from $g^\omega$}.

  Let $r\in\nat_{> 0}$ be minimal such that $g^r$ respects $\Gamma$ ($r$ exists by Lemma~\ref{lem:respect}).
  Let $\dead$ be the set of dead letters with respect to $g^r$ and $\Gamma$.
  For $x \in \Sigma^*$ we use brackets  $[x]$ to denote a new letter.
  For words $w \in \{ g^r(a) \mid a \in \Sigma \}$,
  whenever $\erase{\dead}{w} = w_0\;a_1w_1 \; a_2 w_2\; \cdots\; a_{k-1}w_{k-1}\; a_k w_{k}$
  with $a_1,\ldots,a_{k} \not\in \Gamma$ and $w_0,\ldots,w_{k} \in \Gamma^*$, we define
  \begin{align*}
    \blocks{w} &=
    [w_0a_1w_1] \; [a_2w_2] \;\cdots\; [a_{k-1}w_{k-1}] \; [a_k w_{k}] 
  \end{align*}
  Here it is to be understood that $\blocks{w} = \emptyword$ if $\erase{\dead}{w} = \emptyword$,
  and $\blocks{w}$ is \emph{undefined} if $\erase{\dead}{w} \in \Gamma^+$.

  Let the alphabet $\Delta$ consist of all letters $[a]$ 
  and all bracketed letters $[w]$ occurring in words $\blocks{g^r(a)}$ for $a \in \Sigma$.
  We define the morphism $\xi : \Delta \to \Delta^*$ 
  and the coding $\rho : \Delta \to \Sigma$ by
  \begin{align*}
    \xi([a_1\cdots a_k]) &=  \blocks{g^r(a_1)}\cdots \blocks{g^r(a_k)}
    & \rho([w \,a\, u]) &= a
  \end{align*}
  for $[a_1\cdots a_k] \in \Delta$ and $a\not\in\Gamma$, $w,u \in \Gamma^*$.
  For $a \in \Gamma$ we can define $\rho([a])$ arbitrarily, for example, $\rho(a) = a$.
\end{definition}

\begin{remark}\label{rem:respects}
  The requirement that $g^r$ respects~$\Gamma$ in Definition~\ref{def:close:erase}
  guarantees for every $a \in \Sigma$ that
  either $g^r(a)$ consists of dead letters only 
  or $g^r(a)$ contains at least one near dead or resilient letter. 
  In both cases, $\blocks{g^r(a)}$ is well-defined.
  As a consequence $\xi([w])$ is well-defined for every $[w] \in \Delta$.
\end{remark}

\begin{example}\label{ex:tribonacci:a}
  We let $\Sigma = \{\,a,b,c\,\}$ and define a morphism $g : \Sigma \to \Sigma^*$ by 
  $a \mapsto ab$, $b \mapsto ac$ and $c \mapsto a$.
  The word $g^\omega(a) = abacabaabacababacabaabacabacabaabacababa\cdots$ is known as the \emph{tribonacci word}.
  
  Let $\Gamma = \{\,a\,\}$, that is, we delete the letter $a$.
  The morphism $g$ does not respect $\Gamma$ since $g(c) = a \in \Gamma^*$ but $g^2(c) = ab \not\in \Gamma^*$. 
  However, $g^2$ respects $\Gamma$:
  $g^2(a) = a b a c$,
  $g^2(b) = a b a$ and
  $g^2(c) = a b$.
  The letter $a$ is resurrecting and $b,c$ are resilient with respect to $g^2$ and $\Gamma$.
  Definition~\ref{def:close:erase} yields $\Delta = \{\,[a],[b],[c],[ab],[a b a]\,\}$ and
  \begin{align*}
    \xi([a]) &= [a b a] [c] &
    \xi([b]) &= [a b a] &
    \xi([c]) &= [a b] \\
    \xi([a b]) &= [a b a] [c] [a b a] &
    \xi([a b a]) &= [a b a] [c] [a b a] [a b a] [c]
  \end{align*}
  while 
  $\rho([b]) = \rho([ab]) = \rho([aba]) = b$,
  $\rho([c]) = c$, 
  and $\rho([a]) = a$.
  %
  The starting letter for iterating $\xi$ is $[a]$
  (since the tribonacci word starts with $a$).
  The first iterations of $\xi$ are:
  \begin{align*}
    [a] 
    &\mapsto [a b a] [c]
    \mapsto [a b a] [c] [a b a] [a b a] [c] [a b]\\
    &\mapsto [a b a] [c] [a b a] [a b a] [c]   [a b]   [a b a] [c] [a b a] [a b a] [c]   [a b a] [c] [a b a] [a b a] [c]   [a b]   [a b a] [c] [a b a]
    \\
    &\mapsto\cdots
  \end{align*}
  Then an application of the coding $\rho$ yields $\rho(\xi^\omega([a])) = bcbbcbbcbbcbcbbcbb\cdots = \erase{a}{g^\omega(a)}$.
\end{example}

\begin{example}
  We let $\Sigma = \{\,a,b,c,d,e\,\}$ and define $g : \Sigma \to \Sigma^*$ by 
  $a \mapsto a b c d e$, $b \mapsto c c$, $c \mapsto b$, $d \mapsto c$ and $e \mapsto e a$.
  We let $\Gamma = \{\,b,e\,\}$.
  Then $g^2$ respects $\Gamma$:
  $a \mapsto a b c d e c c b c e a$, $b \mapsto b b$, $c \mapsto c c$, $d \mapsto b$ and $e \mapsto e a a b c d e$.
  Here $b$ is dead, $d$ near dead, $a$ and $c$ are resilient and $e$ is resurrecting.
  Definition~\ref{def:close:erase} yields
  \begin{align*}
    \xi([a]) &= [a] [c] [d e] [c] [c] [c e] [a] &
    \xi([b]) &= \emptyword \quad\quad\quad
    \xi([c]) = [c] [c] \quad\quad\quad
    \xi([d]) = \emptyword \\
    \xi([e]) &= [e a] [a] [c] [d e] &
    \xi([c e]) &= [c] [c] [e a] [a] [c] [d e] \\
    \xi([d e]) &= [e a] [a] [c] [d e] &
    \xi([e a]) &= [e a] [a] [c] [d e] [a] [c] [d e] [c] [c] [c e] [a]
  \end{align*}
  where $\Delta = \{\,[a],[b],[c],[d],[e],[ce],[de],[ea]\,\} $.
  Moreover, we have
  $\rho([a]) = \rho([ea]) = a$, $\rho([c]) = \rho([ce]) = c$, $\rho([d]) = \rho([de]) = d$,
  $\rho([b]) = b$, and $\rho([e]) = e$.
\end{example}

\begin{proposition}\label{prop:close:erase}
  Let $g : \Sigma^*\to \Sigma^*$ be a morphism, 
  $a\in\Sigma$ such that $g^\omega(a) \in \Sigma^\omega$,
  and $\Gamma \subseteq \Sigma$ a set of letters.
  Let $\Delta$, $\xi$ and $\rho$ be the morphic system associated to the erasure of $\Gamma$ from $g^\omega$ in  
  Definition~\ref{def:close:erase}.
  Then 
  \begin{align*}
    \rho(\xi^\omega([a])) = \erase{\Gamma}{g^\omega(a)}
  \end{align*}
\end{proposition}
\begin{proof}
    For $\ell \in \nat$ and $[w_1],\ldots,[w_\ell] \in \Delta$ we define
    $\unblock{[w_1]\cdots[w_\ell]} = w_1\cdots w_\ell$.
    We prove by induction on $n$ that for all words $w \in \Delta^*$, and for all $n \in \nat$,
    $\unblock{\xi^n(w)} = g^{nr}(\unblock{w})$.
    The base case is immediate.
    For the induction step, assume that we have $n\in\nat$ such that 
    for all words $w \in \Delta^*$, $\unblock{ \xi^n(w) } = g^{nr}(\unblock{w})$.
    Let $w \in \Delta^*$, $w = [a_{1,1}\cdots a_{1,\ell_1}]\cdots [a_{k,1}\cdots a_{k,\ell_k}] $. Then 
    \begin{align*}
      \unblock{\xi(w)} 
        &= \unblock{ \xi([a_{1,1}\cdots a_{1,\ell_1}]) \; \cdots \; \xi([a_{k,1}\cdots a_{k,\ell_k}]) } \\
        &= \unblock{ \blocks{g^r(a_{1,1})} \cdots \blocks{g^r(a_{1,\ell_1})}\; \cdots \; \blocks{g^r(a_{k,1})} \cdots \blocks{g^r(a_{k,\ell_k})} }\\
        &= g^r(\unblock{w})
    \end{align*}
    By the induction hypothesis, 
    $\unblock{ \xi^{n+1}(w) }  
        = g^{nr}( \unblock{ \xi(w) } )
        = g^{nr}(g^r(\unblock{w}))
        = g^{(n+1)r}(\unblock{w})$.
    To complete the proof, note that 
    by definition $\rho([w \,a\, u]) = \erase{\Gamma}{w \,a\, u}$
    and thus
    $\rho(w) = \erase{\Gamma}{\unblock{w}}$ for every $w \in \Delta^*$.
    Hence, for all $n\geq 1$, 
    $\rho(\xi^n([a])) = \erase{\Gamma}{\unblock{ \xi^n([a]) }} = \erase{\Gamma}{ g^{nr}(a) }$.
    Taking limits:
    $\rho(\xi^\omega([a])) = \erase{\Gamma}{ g^\omega(a) }$.
\end{proof}

\begin{definition}\normalfont\label{def:close:morphism:nonerasing}
  Let $g,h : \Sigma^*\to \Sigma^*$ be morphisms such that $h$ is non-erasing.
  We construct an alphabet $\Delta$, a morphism $\xi : \Delta^*\to \Delta^*$ 
  and a coding $\rho : \Delta \to \Sigma$ as follows.
  We refer to $\Delta,\xi,\rho$ as the \emph{morphic system associated with the morphic image of $g^\omega$ under $h$}.
    
  Let $\Delta = \Sigma \cup \{\,[a] \mid a\in \Sigma\,\}$.
  For nonempty words $w = a_1 a_2 \cdots a_k \in \Sigma^*$ we define $\head{w} = a_1$ and 
  $\tail{w} = a_2 \cdots a_k$. We also define
  $\img{w} = [a_1] u_1 \; [a_2] u_2 \;\cdots\; [a_{k-1}] u_{k-1} \; [a_k] u_{k}$
  where $u_i = \tail{h(a_i)} \in \Sigma^*$.
  We define the morphism $\xi : \Delta^*\to \Delta^*$ and the coding $\rho : \Delta \to \Sigma$  by
  \begin{align*}
    \xi([a]) &= \img{g(a))} &
    \xi(a) &= \emptyword &&&&&&&
    \rho([a]) &= \head{h(a)} &  
    \rho(a) &= a
  \end{align*}
  for $a\in \Sigma$.
\end{definition}

Notice here the $\rho([a])$ and $u_i$, defined using $\head$ and $\tail$, are well-defined since $h$ is non-erasing and hence $h(a_i)$ will be nonempty.

\begin{example}  Here is an example illustrating Definition~\ref{def:close:morphism:nonerasing}.
  Let $g$ be the substitution from the Fibonacci word, $g(a) = ab$ and $g(b) = a$.
  Further, let $h$ be defined so that $h(a) = bb$ and $h(b) = a$.  
  As in  Definition~\ref{def:close:morphism:nonerasing}, let $\xi$ and $\rho$ be defined by 
  \begin{align*}
    \xi([a]) = [a] b [b] && \xi([b]) = [a]b && \xi(a) = \emptyword = \xi(b) &&&&&& \rho([a]) = b && \rho([b]) = a
  \end{align*}
  Then 
  $[a] 
    \mapsto [a]b[b] 
    \mapsto [a]b[b] [a]b 
    \mapsto [a]b[b] [a]b [a]b[b]
    \mapsto [a]b[b] [a]b [a]b[b] [a]b[b] [a]b
    \mapsto\cdots$
  are the first iterations of $\xi$ on $[a]$.
%
  The point here is that applying $\rho$ to the limit word $\xi^\omega([a])$
  is the same as $h(g^\omega(a))$:  
  \begin{align*}
    h(g^\omega(a)) 
    = h(abaababaabaababaabab\cdots) 
    = bbabbbbabbabbbbabbbbabb\cdots 
  \end{align*}
\end{example}

\begin{proposition}\label{prop:close:morphism}
  Let $g,h : \Sigma^*\to \Sigma^*$ be morphisms such that $h$ is non-erasing,
  and $a\in\Sigma$ such that $g^\omega(a) \in \Sigma^\omega$.
  Let $\Delta$, $\xi$ and $\rho$ 
 be as in Definition~\ref{def:close:erase}.
  Then
  \begin{align*}
    \rho(\xi^\omega([a])) = h(g^\omega(a))
  \end{align*}
\end{proposition}
\begin{proof}
  We define $z : \Delta \to \Sigma^*$ by
  $z(a) = \emptyword$ and $z([a]) = a$ for all $a\in\Sigma$.
  By induction on $n > 0$ we show 
  \begin{align}
    \rho(\xi^n(w)) = h(g^n(z(w))) &&\text{and}&& z(\xi^n(w)) = g^n(z(w)) && \quad{\text{for all $w\in \Delta^*$}}
  \end{align}
  We start with the base case.
  Note that $\rho(\xi([a])) = h(g(a)) = h(g(z([a])))$
  and $\rho(\xi(a)) = \emptyword = h(g(z(a)))$ for all $a\in \Sigma$,
  and thus $\rho(\xi(w)) = h(g(z(w)))$ for all $w\in \Delta^*$.
  Moreover, we have
  $z(\xi([a])) = g(a) = g(z([a]))$
  and $z(\xi(a)) = \emptyword = g(z(a))$ for all $a\in \Sigma$,
  and hence $z(\xi(w)) = g(z(w))$ for all $w\in \Delta^*$.

  Let us consider the induction step.
  By the base case and induction hypothesis
  \begin{align*}
    \rho(\xi^{n+1}(w)) &=  \rho(\xi(\xi^{n}(w))) = h(g(z(\xi^{n}(w)))) = h(g(g^n(z(w)))) = h(g^{n+1}(z(w)))\\
    z(\xi^{n+1}(w)) &=  z(\xi(\xi^{n}(w))) = g(z(\xi^{n}(w))) = g(g^n(z(w))) = g^{n+1}(z(w))
  \end{align*}
  Thus $\rho(\xi^{n}([a])) = h(g^n(a))$ for all $n\in\nat$,
  and taking limits yields $\rho(\xi^\omega([a])) = h(g^\omega(a))$.
\end{proof}

Every morphic image of a word can be obtained by erasing letters,
followed by the application of a non-erasing morphism.
As a consequence we obtain:
\begin{corollary}\label{cor:close:pure} 
  The morphic image of a pure morphic word is morphic or finite.
\end{corollary}

\begin{proof}
  Let $w \in \Sigma^\omega$ be a word and $h : \Sigma \to \Sigma^*$ a morphism.
  Let $\Gamma = \{\;a \mid h(a) = \emptyword\;\}$ be the set of letters erased by $h$,
  and $\Delta = \Sigma\setminus\Gamma$.
  Then $h(w) = g(\erase{\Gamma}{w})$
  where $g$ is the non-erasing morphism obtained by restricting $h$ to $\Delta$.
  Hence for purely morphic $w$, the result follows from Propositions~\ref{prop:close:erase} and~\ref{prop:close:morphism}.
\end{proof}

\begin{theorem}  [Cobham~\cite{cobh:68}, Pansiot~\cite{pans:83}]
  The morphic image of a morphic word is morphic.
\end{theorem}
\begin{proof}
  Follows from Corollary~\ref{cor:close:pure}
  since the coding can be absorbed into the morphic image. 
\end{proof}

\subsection*{Eigenvalue analysis}

The following lemma states that if a square matrix $N$ is an extension of a square matrix $M$,
and all added columns contain only zeros, then $M$ and $N$ have the same non-zero eigenvalues.

\begin{center}
 $\left(
  \begin{array}{cccc}
    \begin{array}{ccccc|}
       &      & &  & \\[-1.5ex]
  &     M & &  &\\  [-2ex]
  &      & &  & \\
      \hline      
    \end{array}
  & 0 & \cdots & 0  \\
  & 0& \cdots & 0 \\[-2ex]
  \\   
  & 0& \cdots & 0 \\ 
  \end{array}
\right)  $
\end{center}
\smallskip

\begin{lemma}\label{lem:matrix:zero}
  Let $\Sigma$, $\Delta$ be disjoint, finite alphabets.
  Let $M = (m_{i,j})_{i,j\in \Sigma}$ and $N = (n_{i,j})_{i,j\in \Sigma\cup \Delta}$
  be matrices such that
  (i) $n_{i,j} = m_{i,j}$ for all $i,j\in \Sigma$ and 
  (ii) $n_{i,j} = 0$ for all $i \in \Sigma\cup \Delta$, $j \in \Delta$.
  Then $M$ and $N$ have the same non-zero eigenvalues.
\end{lemma}

\begin{proof}
$N$ is a block lower triangular matrix with $M$ and $0$ as the matrices on the diagonal.
Hence the eigenvalues of $N$ are the combined eigenvalues of $M$ and $0$. Therefore $M$ and $N$
have the same non-zero eigenvalues.
\end{proof}

We now show that morphic images with respect to non-erasing morphisms preserve $\alpha$-substitutivity.
This strengthens a result obtained in~\cite{dura:2011}
where it has been shown that the non-erasing morphic image of an $\alpha$-substitutive sequence
is $\alpha^k$-substitutive for some $k\in \nat$.
We show that one can always take $k = 1$.
Note that every $\alpha$-substitutive sequence is also $\alpha^k$-substitutive for all $k\in\nat, k > 0$.

\begin{theorem}\label{thm:morphism}
  Let $\aalph$ be a finite alphabet, $w \in \aalph^\omega$ be an $\alpha$-substitutive sequence and $h : \aalph \to \aalph^*$ a non-erasing morphism.
  Then the morphic image of $w$ under $h$, that is $h(w)$, is $\alpha$-substitutive.
\end{theorem}

\begin{proof}
  Let $\aalph = \{\,a_1,\ldots,a_k\,\}$ be a finite alphabet, $w \in \aalph^\omega$ be an $\alpha$-substitutive sequence and $h : \aalph \to \aalph^*$ a non-erasing morphism.
  As the sequence $w$ is $\alpha$-substitutive, there exist
  a morphism $g : \aalph \to \aalph^*$ with dominant eigenvalue~$\alpha$,
  a coding $c : \aalph \to \aalph$
  and a letter $a \in \aalph$
  such that 
  $w = c(g^\omega(a))$ and all letters from $\aalph$ occur in $g^\omega(a)$.
  Then $h(w) = h(c(g^\omega(a))) = (h \circ c)(g^\omega(a)))$,
  and $h \circ c$ is a non-erasing morphism.
  Without loss of generality, by absorbing $c$ into $h$, we may assume that $c$ is the identity.
  
  From $h$ and $g$, we obtain 
  an alphabet $\Delta$,  a morphism $\xi$,  and a coding $\rho$ as in Definition~\ref{def:close:morphism:nonerasing}.
  Then by Proposition~\ref{prop:close:morphism},
  we have $\rho(\xi^\omega([a])) = h(g^\omega(a))$.
  As a consequence, it suffices to show that $\rho(\xi^\omega([a]))$ is $\alpha$-substitutive.
  Let $M = (M_{i,j})_{i,j \in \Sigma}$ and $N = (N_{i,j})_{i,j \in \Delta}$ 
  be the incidence matrices of $g$ and $\xi$, respectively.
  By Definition~\ref{def:close:morphism:nonerasing} we have for all $a,b\in\aalph$:
  $\occ{[b]}{\xi([a])} = \occ{b}{g(a)}$ and
  $\occ{b}{\xi(a)} = \occ{[b]}{\xi(a)} = 0$.
  Hence we obtain
  $N_{[b],[a]} = M_{b,a}$, 
  $N_{b,a} = 0$
  and
  $N_{[b],a} = 0$ 
  for all $a,b\in\aalph$.
  After changing the names (swapping $a$ with $[a]$) in $N$,
  we obtain from Lemma~\ref{lem:matrix:zero}
  that $N$ and $M$ have the same non-zero eigenvalues,
  and thus the same dominant eigenvalue.
\end{proof}

\begin{example}
  Let $F$ be the Fibonacci word (generated by the morphism $a \mapsto ab$ and $b \mapsto a$)
  and let $T$ be the \thuemorse{} sequence.
  We show that there exist no non-erasing morphisms $g,h$
  such that $g(F) = h(T)$ and this image is not ultimately periodic.
  Let $g$ and $h$ be non-erasing morphisms.
  The Fibonacci word is $\phi$-substitutive where $\phi = (1+\sqrt{5})/2$ is the golden ratio,
  and the Thue-Morse sequence is $2$-substitutive.
  By Theorem~\ref{thm:morphism}, $g(F)$ is $\phi$-substitutive and $h(T)$ is $2$-substitutive.
  Note that $\phi$ and $2$ are multiplicatively independent:
  using induction on $k \in \nat_{>0}$
  it follows that every $\phi^k$ is of the form $a + b\sqrt{5}$ for rational numbers $a,b > 0$.
  It follows by Theorem~\ref{theorem-Durand} that $g(F) = h(T)$ implies that this word is ultimately periodic.
\end{example}

\begin{remark}\label{rem:erasing:morphism}
  The restriction to non-erasing morphisms in Theorem~\ref{thm:morphism} is important
  since every morphic sequence can be obtained by erasure of letters from a $2$-substitutive sequence.

  Nevertheless, we can use the above theorem to reason about morphic images with respect to erasing morphisms as follows.
  Let $w \in  \aalph^\omega$, and $g : \aalph \to \aalph^*$ a morphism.
  Let $\balph$ be the letters erased by $g$, and let $h$ be the restriction of $g$ to $\aalph \setminus \balph$.
  Then $h$ is non-erasing and $g(w) = h(\erase{\balph}{w})$.
  Hence, if $\erase{\balph}{w}$ is $\alpha$-substitutive, then so is $g(w)$ by Theorem~\ref{thm:morphism}.
  As a consequence, it suffices to determine $\alpha$-substitutivity of all sequences $\erase{\balph}{w}$ with $\balph \subseteq \aalph$
  (using Definition~\ref{def:close:erase} and Proposition~\ref{prop:close:erase}).
\end{remark}

\section{Closure of Morphic Sequences under Transduction}
\label{section-closure-transduction}

\def\d{\delta}
\def\e{\epsilon}

\label{section-Dekking}
In this section, we give a proof of  the following theorem due to Dekking~\cite{dekk:94}.

\begin{theorem}[Transducts of morphic sequences are morphic]\label{the:transducts_preserve_morphic}
  If $M = (\Sigma, \Delta, Q, q_0, \delta, \lambda)$ is a transducer with input alphabet $\Sigma$ 
  and $x \in \Sigma^\omega$ is a morphic sequence, then $M(x)$ is morphic or finite.
\end{theorem}

This proof will proceed by \emph{annotating} entries in the original sequence $x$ 
with information about what state the transducer is in upon reaching that entry.
This allows us to construct a new morphism which produces the transduced sequence $M(x)$ as output. 
After proving this theorem, we will show that this process of annotation preserves $\alpha$-substitutivity.

\begin{figure}[h!]
  \centering
  \includegraphics{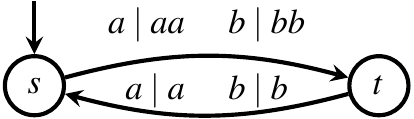}
  \caption{A transducer that doubles every other letter.}
  \label{fig:double}
\end{figure}

\begin{example}\label{example-transducer}
  To illustrate several points in this section,
  we will consider the Fibonacci morphism ($h(a) = ab$, $h(b) = a$) 
  and the transducer which doubles every other letter, shown in Figure~\ref{fig:double}. 
\end{example}

\subsection{Transducts of morphic sequences are morphic}


We show in Lemma~\ref{lem:state_annotated_is_morphic} that transducts of morphic sequences
are morphic.
In order to prove this, we  also need several lemmas about transducers which are of independent interest. 
The approach here is adapted from a result in Allouche and Shallit~\cite{allo:shal:2003}; it is  attributed in that book  to Dekking. 
We repeat it here partly for the convenience of the reader, but mostly because there are some details of the proof which 
are used in the analysis of the substitutivity property.

\begin{definition}[$\tau_w$, $\Xi(w)$]
  Given a transducer $M = (\Sigma, \Delta, Q, q_0, \delta, \lambda)$ and a word $w \in \Sigma^*$,
  we define $\tau_w \in Q^Q$ to be $\tau_w(q) = \d(q, w)$. Note that $\tau_{wv} = \tau_v \circ \tau_w$.
  Further, we define $\Xi: \Sigma^*\to  (Q^Q)^\omega$ by
  $\Xi(w) = (\tau_w, \tau_{h(w)}, \tau_{h^2(w)}, \ldots, \tau_{h^n(w)}, \ldots)$.
  \label{definition-tau}
\end{definition}

\begin{example}\label{ann_ex:tau_and_theta'}
  Recall the transducer $M$ from Figure~\ref{fig:double}. 
  Let $id: Q \to Q$ be the identity, 
  and let $\nu: Q \to Q$ be the transposition $\nu(s) = t$ and $\nu(t) = s$. 
  For this transducer, 
  $\tau_w = id$ if  $|w|$ is even and $\tau_w = \nu$ if  $|w|$ is odd.
  We have $\Xi(a) = (\tau_a, \tau_{ab}, \tau_{aba}, \tau_{abaab}, \tau_{abaababa}, \ldots)$.
  In this notation,
  \begin{align*}
    \Xi(a) & = (\nu, id, \nu, \nu, id, \nu, \nu, id, \nu, \nu, \ldots) &
    \Xi(b) & =  (\nu, \nu, id, \nu, \nu, id, \ldots) &
    \Xi(\e) & =  (id, id, id, id, \ldots)
  \end{align*}
\end{example}

Next, we show that $\set{\,\Xi(w) : w \in \Sigma^*\,}$ is finite.

\begin{lemma}\label{lem:finite_annotation_exists}
  For any  transducer $M$ and any morphism $h:\Sigma\to\Sigma^*$,  
  there are natural numbers $p \geq 1$ and $n \geq 0$ so that for all $w \in \Sigma^*$, 
  $\tau_{h^i(w)} = \tau_{h^{i+p}(w)}$ for all $i \geq n$.
\end{lemma}

\begin{proof}
  Let $\Sigma = \{1, 2, \ldots, s\}$. Define $H: (Q^Q)^{s} \to (Q^Q)^{s}$ by 
  $H(f_1, f_2, \ldots, f_{s}) =  (f_{h(1)}, f_{h(2)}, \ldots, f_{h(s)})$.
  When we write $f_{h(i)}$ on the right, here is what we mean.  Suppose that $h(i) = v_0\cdots v_j$.
  Then $f_{h(i)}$ is short for the composition $f_{v_j} \circ f_{v_{j-1}} \circ \cdots \circ f_{v_1} \circ f_{v_0}$.
  Recall the notation $\tau_w$ from Definition~\ref{definition-tau}; we thus have $\tau_i$ for  the individual letters $i\in\Sigma$.
  Consider $T_0 = (\tau_1, \tau_2, \ldots, \tau_{s})$.   We define its \emph{orbit} as 
  the infinite sequence $(T_i)_{i\in\omega}$ of elements of $ (Q^Q)^{s}$ given by
  $T_i =  H^i(T_0) =  H^i(\tau_1, \ldots \tau_{s}) =  (\tau_{h^i(1)}, \ldots, \tau_{h^i(s)})$.
  Since each of the $T_i$ belongs to the finite set $(Q^Q)^{s}$,  the orbit of $T_0$ is eventually periodic. 
  Let $n$ be the preperiod length and $p$ be the period length. The
  periodicity implies that $(*)$ $\tau_{h^i(j)} = \tau_{h^{i+p}(j)}$ for each $j \in \Sigma$ and for all $i \geq n$.

  Let $w \in \Sigma^*$ and $i \geq n$. Since $w \in \Sigma^*$, we can write it as $w = \sigma_1\sigma_2 \cdots \sigma_m$.
  We prove that 
  $\tau_{h^i(w)} = \tau_{h^{i+p}(w)}$. 
  Note that
  $\tau_{h^i(w)} =  \tau_{h^i(\sigma_1 \cdots \sigma_m)} =  \tau_{h^i(\sigma_1)h^i(\sigma_2) \cdots h^i(\sigma_m)}
  =  \tau_{h^i(\sigma_n)}\circ\cdots\circ\tau_{h^i(\sigma_1)}$.
  We got this  by breaking $w$ into individual letters, then using the fact that $h$ is a morphism, and finally using 
  the fact that $\tau_{u v} = \tau_u\circ\tau_v$.
  Finally we know by $(*)$ that for individual letters,
  $\tau_{h^i(\sigma_j)} = \tau_{h^{i+p}(\sigma_j)}$.
  So  $\tau_{h^i(w)} = \tau_{h^{i+p}(w)}$, as desired.
\end{proof}

\begin{definition}[$\Theta(w)$]
  Given a  transducer $M$ and a morphism $h$, we find $p$ and $n$ as in Lemma~\ref{lem:finite_annotation_exists}
  just above and define $\Theta(w) = (\tau_w, \tau_{h(w)}, \ldots, \tau_{h^{n+p-1}(w)})$.
\label{def-Theta}
\end{definition}

\begin{example}  We continue with Example~\ref{example-transducer}.
  As the proof in Lemma~\ref{lem:finite_annotation_exists} demonstrates, to find the $p$ and $n$ for our transducer and the
   Fibonacci morphism, we only need to  find the common period of  $\Xi(a)$ and $\Xi(b)$.
Using what we saw in Example~\ref{ann_ex:tau_and_theta'} above,  we can  take $n = 0$ and $p = 3$. 
  Therefore, $\Theta(a) = (\nu, id, \nu)$ and $\Theta(b) = (\nu, \nu, id)$. 
  We also note that $\Theta(\e) = (id, id, id)$ and $\Theta(ab) = (id, \nu, \nu)$, as we will need these later.
\end{example}

\begin{lemma}
\begin{enumerate}[(i)]
\item Given $M$ and $h$, the set $A = \{\, \Theta(w) : w \in \Sigma^*\,\}$
  is finite.
  \item 
 If   $\Theta(w) = \Theta(y)$, then $\Theta(h(w)) = \Theta(h(y))$.
\item  If   $\Theta(w) = \Theta(y)$, then for all $u\in\Sigma^*$,  $\Theta(wu) = \Theta(y u)$.
\end{enumerate}
\label{lemma-for-welldefinedness}
\end{lemma}

\begin{proof}
Part (i) comes from the fact
that  each of the $n+p$ coordinates of $\Theta(w)$ comes from the finite set $Q^Q$.
For (ii), we calculate:
$$\begin{array}{lcl@{\qquad}l}
  \Theta(h(w)) & = & (\tau_{h(w)}, \tau_{h^2(w)}, \tau_{h^3(w)}, \ldots ,\tau_{h^{n+p}(w)}) \\
  &=& (\tau_{h(w)}, \tau_{h^2(w)}, \tau_{h^3(w)}, \ldots,  \tau_{h^{n+p-1}(w)} ,\tau_{h^{n}(w)})  
& \mbox{by Lemma~\ref{lem:finite_annotation_exists}}
\\
 & = & (\tau_{h(y)}, \tau_{h^2(y)}, \tau_{h^3(y)}, \ldots,   \tau_{h^{n+p-1}(y)}, \tau_{h^{n}(y)}) = \Theta(h(y))
 & \mbox{since $\Theta(w) = \Theta(y)$} 
\end{array}
$$
Part (iii) uses $\Theta(w) = \Theta(y)$ as follows:
$$\begin{array}{lcl@{\qquad}l}
\Theta(w u) & = & (\tau_{wu}, \tau_{h(wu)}, \tau_{h^2(wu)}, \ldots ,\tau_{h^{n+p-1}(wu)})  \\
  & = & (\tau_u\circ \tau_{w},  \tau_{h(u)}\circ \tau_{h(w)}, \tau_{h^2(u)}\circ \tau_{h^2(w)}, \ldots ,\tau_{h^{n+p-1}(u)}\circ \tau_{h^{n+p-1}(w)})  \\
  & = & (\tau_u\circ \tau_{y},  \tau_{h(u)}\circ \tau_{h(y)}, \tau_{h^2(u)}\circ \tau_{h^2(y)}, \ldots ,\tau_{h^{n+p-1}(u)}\circ \tau_{h^{n+p-1}(y)}) 
  = \Theta(y u)
\end{array}
$$
\end{proof}


\begin{definition}[$\overline{h}$]\label{Def-hbar}
  Given a  transducer $M$ and a morphism $h$, 
  let $A$ be as in Lemma~\ref{lemma-for-welldefinedness}(i).  Define
  the morphism $\overline{h}: \Sigma \times A \to (\Sigma \times A)^*$ as follows.
  For for all $\sigma \in \Sigma$, whenever $h(\sigma) = s_1s_2s_3\cdots s_\ell$, let
  $$\overline{h}((\sigma, \Theta(w))) \;\;=\;\; 
    (s_1, \Theta(h w))\;\; (s_2, \Theta((hw)s_1))
    \;\; (s_3, \Theta((hw)s_1s_2))
    \;\; \cdots \;\; (s_{\ell}, \Theta((h w)s_1s_2\cdots s_{\ell-1}))$$
\end{definition}
By repeated use of Lemma~\ref{lemma-for-welldefinedness}, $\overline{h}$  is well-defined.
Notice that $|\overline{h}(\sigma, a)| = |h(\sigma)|$ for all $\sigma$. 

\begin{lemma}\label{lem:hannotate}
  For all $\sigma\in\Sigma$, all $w\in \Sigma^*$
  and all natural numbers $n$, if $h^n(\sigma) = s_1 s_2 \cdots s_{\ell}$,
  then 
  $$\overline{h}^n((\sigma, \Theta(w))) \;\;=\;\;
  (s_1, \Theta(h^n w))\;\; (s_2, \Theta((h^n w)s_1)) \;\; (s_3, \Theta((h^n w)s_1s_2)) \;\cdots\; (s_{\ell}, \Theta((h^n w)s_1s_2\cdots s_{\ell-1})).$$ 
  In particular, for $1\leq i \leq \ell$, the first component of the $i^{th}$ term
  in $h^n(\sigma,\Theta(w))$ is $s_i$. 
\end{lemma}

\begin{proof} By induction on $n$.
  For $n = 0$, the claim is trivial.   Assume that it holds for $n$.
  Let $h^n(\sigma) = s_1 s_2 \cdots s_{\ell}$,
  and  for $1\leq i \leq \ell$,
  let $h(s_i) = t^i_1 t^i_2 \cdots t^i_{k_i}$.  Thus
  $
  h^{n+1}(\sigma) =
  h(s_1 s_2\cdots s_{\ell}) = t^1_1 t^1_2 \cdots t^1_{k_1} t^2_1 t^2_2 \cdots t^2_{k_2}
  t^\ell_1 t^\ell_2 \cdots t^\ell_{k_\ell}
  $.
  Then:
  $$\begin{array}{clcl}
  & \overline{h}(\overline{h}^n(\sigma, \Theta(w))) 
  \;=\; \hbar(s_1, \Theta((h^n w)))\;\; \hbar(s_2, \Theta((h^n w)s_1))
  \;\; \hbar(s_3, \Theta((h^n w)s_1s_2))\\
  & \hspace{1.5in}\cdots\;\; \hbar(s_{\ell},
   \Theta((h^n w)s_1s_2\cdots s_{\ell-1}))
  \end{array}
  $$
  For $1\leq i \leq \ell$, we have
  $$\begin{array}{clcl}
   & \hbar(s_i, \Theta((h^n w)s_1\cdots s_{i-1})) \\
   = &   (t^i_1, \Theta((h h^n w)h(s_1\cdots s_{i-1})))
   \quad (t^i_2, \Theta((hh^nw)h(s_1\cdots s_{i-1})t^i_1))
  \\
   & \hspace{.5in}\cdots\quad 
  (t^i_{k_i}, \Theta(hh^nw)h(s_1\cdots s_{i-1}) t^i_1 t^i_2\cdots t^i_{k_i-1}))  \\
  = & (t^i_1, \Theta((h^{n+1}w) t^1_1 t^1_2 \cdots t^1_{k_1} \cdots t^{i-1}_1 t^{i-1}_2 \cdots t^{i-1}_{k_{i-1}}))
  \quad 
  (t^i_2, \Theta((h^{n+1}w) t^1_1 t^1_2 \cdots
   t^1_{k_1} \cdots t^{i-1}_1 t^{i-1}_2 \cdots t^{i-1}_{k_{i-1}} t^i_1))
   \\
   & \hspace{.5in} \cdots \quad 
   (t^i_{k_i}, \Theta((h^{n+1}w) t^1_1 t^1_2 \cdots
   t^1_{k_i} \cdots t^{i-1}_1 t^{i-1}_2 \cdots t^{i-1}_{k_{i-1}} t^i_1 \cdots t^i_{k_i -1}))
  \end{array}
  $$
  Concatenating the  sequences $\hbar(s_i, \Theta((h^n w)s_1\cdots s_{i-1}))$
  for $i = 1, \ldots, \ell$ completes our induction step.
\end{proof}

\begin{lemma}\label{lem:state_annotated_is_morphic}
  Let $M = (\Sigma, \Delta, Q, q_0, \delta, \lambda)$ be a  transducer, 
  let $h$ be a morphism prolongable on the letter $x_1$,   
  and write $h^\omega(x_1)$ as  $x = x_1x_2x_3\cdots x_n \cdots$.
 Let $\Theta$ be from Definition~\ref{def-Theta}.  Using this,
 let  $A$ be  from
 Lemma~\ref{lemma-for-welldefinedness}(i), and
  $\overline{h}$ from Definition~\ref{Def-hbar}.    Then
  \begin{enumerate}
  \item[(i)]  $\overline{h}$ is prolongable on  
$(x_1, \Theta(\e))$. 
  \item[(ii)]   Let   $c: \Sigma \times A \to \Sigma \times Q$ be the coding 
$c(\sigma, \Theta(w)) = (\sigma, \tau_w(q_0))$.
    Then $c$ is well-defined.
\item[(iii)]  The image under $c$ of $\overline{h}^\omega((x_1, \Theta(\e))$
is
  \begin{equation}
    \begin{array}{lcl}
  z & \;\;=\;\; & (x_1, \d(q_0, \e))\;\;
  (x_2, \d(q_0, x_1))\;\;
  (x_3, \d(q_0, x_1x_2))\;\;\cdots\;\; (x_n, \d(q_0, x_1x_2\cdots x_{n-1}))
  \;\;\cdots
  \end{array}
  \label{eq-z}
  \end{equation}
This sequence $z$ is morphic in the alphabet $\Sigma \times Q$.
  \end{enumerate}
\end{lemma}

\begin{proof}
For (i),   write $h(x_1)$ as $ x_1 x_2 \cdots x_\ell$.
Using the fact that $h^i(\e) = \e$ for all $i$, we see that
$$ 
\begin{array}{lcl}
\overline{h}((x_1, \Theta(\e))) & = & (x_1, \Theta(\e)) 
\quad
  (x_2, \Theta(x_1)) \quad \cdots 
  \quad (x_{\ell}, \Theta(x_1,\ldots, x_{\ell -1})).
\end{array}
$$
This verifies the prolongability.

For (ii): if $\Theta(w) = \Theta(u)$,
  then $\tau_w$ and $\tau_u$ are the first component of $\Theta(w)$ and are thus equal.

We turn to (iii).
Taking $w = \e$ in Lemma~\ref{lem:hannotate} shows that  $\overline{h}^\omega((x_1, \Theta(\e))$ is
$$(x_1, \Theta(\e)) \quad (x_2, \Theta(x_1))
\quad (x_3, \Theta(x_1 x_2)) \quad \cdots \quad
(x_m, \Theta(x_1 x_2\cdots x_{m-1})) \quad \cdots .
$$
The image of this sequence under the coding $c$ is 
$$\begin{array}{lcl}
  (x_1, \tau_{\e}(q_0)) \quad (x_2, \tau_{x_1}(q_0))
  \quad (x_3, \tau_{x_1 x_2}(q_0))\quad \cdots\quad
  (x_m, \tau_{x_1 x_2\cdots x_{m-1}}(q_0)) \quad \cdots .
\end{array}
$$
In view of the  $\tau$ functions' definition (Def.~\ref{definition-tau}), we obtain $z$ in
 (\ref{eq-z}).
By definition,   $z$ is   morphic. 
\end{proof}

This is most of the work required to prove
Theorem~\ref{the:transducts_preserve_morphic},
 the main result of this section. 
 
\begin{proof}[Theorem~\ref{the:transducts_preserve_morphic}]  
  Since $x$ is morphic there is a morphism $h: \Sigma' \to (\Sigma')^*$, a coding $c: \Sigma' \to \Sigma$, 
  and an initial letter $x_1 \in \Sigma'$ so that $x = c(h^\omega(x_1))$.   We are to show that $M(c(h^\omega(x_1)))$
  is morphic.    Since $c$ is computable by a transducer, we have
  $x = (M \circ c)(h^\omega(x_1))$, where $\circ$ is the composition of transducers from Definition~\ref{definition-composition-transducers}.
  It is thus sufficient to show that given a transducer $M$, the sequence $M(h^\omega(x_1))$ is morphic.
  
  By Lemma~\ref{lem:state_annotated_is_morphic}, the sequence 
  $$z = (x_1, \d(q_0, \e)) \quad (x_2, \d(q_0, x_1))\quad
  (x_3, \d(q_0, x_1x_2))\quad \cdots\quad (x_n, \d(q_0, x_1x_2\cdots x_{n-1}))\quad\cdots$$ is morphic. 
  The output function of $M$ is a morphism  $\lambda: \Sigma\times Q \to \Delta^*$.  By 
  Corollary~\ref{cor:close:pure},
  $\lambda(z)$ is morphic or finite.   But  $\lambda(z)$ is exactly $M(x)$; indeed, the definition of $M(x)$ is  
  basically the same as the definition of $\lambda(z)$.  This proves the theorem.
\end{proof}

\subsection{Substitutivity of transducts}

We are also interested in analyzing the $\alpha$-substitutivity of transducts. 
We claim that if a sequence $x$ is $\alpha$-substitutive, then $M(x)$ is also $\alpha$-substitutive for all $M$. 

As a first step, we show that annotating a morphism does not change $\alpha$-substitutivity.
\begin{definition}
  Let $\Sigma$ be an alphabet and $A$ any set.
  Let $w \; = \; (b_1,a_1) \;\; (b_2,a_2) \;\;  \ldots \;\;  (b_k,a_k) \in (\Sigma \times A)^*$ be a word.
  We call $A$ the \emph{set of annotations}.
  We write $\floor{w}$ for the word $b_1b_2\ldots b_k$,
  that is, the word obtained by \emph{dropping the annotations}.

  A morphism $\overline{h} : (\Sigma \times A) \to (\Sigma \times A)^*$
  is an \emph{annotation} of $h : \Sigma \to \Sigma^*$
  if $h(b) = \floor{\overline{h}(b,a)}$ for all $b\in\Sigma$, $a \in A$.
\end{definition}

Note that the morphism $\overline{h}$ from Definition~\ref{Def-hbar} is an annotation of $h$ in this sense.
Then from the following proposition it follows that if $x$ is $\alpha$-substitutive,
then the sequence $z$ in Lemma~\ref{lem:state_annotated_is_morphic} is also $\alpha$-substitutive. 

\begin{proposition}
If $x = h^\omega(\sigma)$ is an $\alpha$-substitutive morphic sequence with morphism $h: \Sigma \to \Sigma^*$ and $A$ is any set of annotations, then any annotated morphism $\overline{h}: \Sigma \times A \to (\Sigma \times A)^*$ also has an infinite fixpoint $\overline{h}^\omega((\sigma,a))$ which is also $\alpha$-substitutive.
\end{proposition}

The proof of this proposition is in two lemmas: first that the eigenvalues of 
the morphism are preserved by the annotation process, and second that if $\alpha$ is the
 dominant eigenvalue for $h$, then no greater eigenvalues are introduced for $\overline{h}$.

\begin{lemma}\label{lem:h:oh}
  All eigenvalues for $h$ are also eigenvalues for any annotated version $\overline{h}$ of $h$.
\end{lemma}

\begin{proof}
  Let $M = (m_{i,j})_{i,j \in \Sigma}$ be the incidence matrix of $h$. 
  Order the elements of the annotated alphabet $\Sigma \times A$ lexicographically. 
  Then the incidence matrix of $\overline{h}$, call it $N = (n_{i,j})_{i,j \in \Sigma \times A}$, can be thought of 
  as a block matrix where the blocks have size $|A| \times |A|$ and there are $|\Sigma| \times |\Sigma|$ such blocks in $N$. 
  Note that by the definition of annotation, the row sum in each row of the $(a,b)$ block of $N$ is $m_{a,b}$.
  To simplify the notation, for the rest of this proof we write $J$ for $|\Sigma|$ and $K$ for $|A|$.
  Suppose $v = (v_1, v_2, \ldots, v_{J})$ is a column eigenvector for $M$ with eigenvalue $\alpha$. 
  Consider
  $\overline{v} = (v_1, \ldots, v_1, v_2, \ldots, v_2, \ldots, v_n, \ldots v_n)$. 
  This  is a ``block vector":  the first $K$ entries are $v_1$, the second $K$ entries are $v_2$, and so on,
  for a total of $K\cdot J$ entries. We claim that $\overline{v}$ is a column eigenvector for $N$ with eigenvalue $\alpha$.

  Consider the product of row $k$ of $N$ with $\overline{v}$. This is 
  $\sum_{j = 1}^{K\cdot J} n_{k,j}\overline{v}_j =  \sum_{b = 1}^{J} v_b\cdot(\sum_{j=1}^{K}n_{k,Kb + j})$.
  Now  $k = Ka + r$.
  So $\sum_{j=1}^{K} n_{k, Kb + j}$ is the row sum of the $(a, b)$ block of $N$ and hence is $m_{a,b}$. Therefore, row $k$ of $N$ times $\overline{v}$ is $\sum_{b = 1}^{J} v_b m_{a, b} = \alpha v_a$,
  since $v$ is an eigenvector of $M$. Finally we note that 
  the $k$th entry of $\overline{v}$ is $v_a$ by its definition. Hence multiplying $\overline{v}$ by $N$ multiplies the $k$th entry of $\overline{v}$ by $\alpha$ for all $k$. 

  We have shown  that $\overline{v}$ is a column eigenvector of $N$ with eigenvalue $\alpha$, so the (column) eigenvalues of $M$ are all present in $N$. However, since a matrix and its transpose have the same eigenvalues, the (column) qualification on the eigenvalues is unnecessary.
  \qed
\end{proof}

If $\overline{h}$ is an annotation of $h$, then we have 
\begin{align}
  \occ{b'}{h(b)} = \sum_{a'\in A} \occ{\pair{b'}{a'}}{\,\overline{h}(\pair{b}{a})\,}
  &&\text{for all $b,b'\in\Sigma$ and $a \in A$} \label{eq:sub:sum}
\end{align}

\begin{lemma}\label{lem:oh:h}
  Let $h,\overline{h}$ be morphisms such that $\overline{h} : (\Sigma \times A) \to (\Sigma \times A)^*$
  is an annotation of $h : \Sigma \to \Sigma^*$.
  Then every eigenvalue of $\overline{h}$ with a non-negative eigenvector is also an eigenvalue for $h$.
\end{lemma}

\begin{proof}
  Let $M = (m_{i,j})_{i,j \in \Sigma}$ be the incidence matrix of $h$
  and $N = (n_{i,j})_{i,j \in \Sigma\times A}$ be the incidence matrix of $\overline{h}$.
  Let $r$ be an eigenvalue of $N$ with corresponding eigenvector $v = (v_{\pair{b}{a}})_{\pair{b}{a} \in \Sigma\times A}$,
  that is, $Nv = rv$ and $v\ne 0$.
  We define a vector $w = (w_b)_{b \in \Sigma}$ as follows:
  $w_b = \sum_{a \in A} v_{\pair{b}{a}}$.
  We show that $Mw = rw$.
  Let $b' \in \Sigma$, then:
  \begin{align*}
    (Mw)_{b'} 
    &= \sum_{b \in \Sigma} M_{b',b}w_b 
    = \sum_{b \in \Sigma} \left( M_{b',b} \sum_{a \in A} v_{\pair{b}{a}} \right)\\
    &= \sum_{b \in \Sigma} \sum_{a \in A} M_{b',b} v_{\pair{b}{a}} 
    \stackrel{\text{by \eqref{eq:sub:sum}}}{=} \sum_{b \in \Sigma} \sum_{a \in A} \left(\sum_{a' \in A} N_{\pair{b'}{a'},\pair{b}{a}}\right) v_{\pair{b}{a}} 
    \\
    &= \sum_{a' \in A} \sum_{b \in \Sigma} \sum_{a \in A} N_{\pair{b'}{a'},\pair{b}{a}} v_{\pair{b}{a}} 
    \stackrel{Nv = rv}{=} \sum_{a' \in A} r v_{\pair{b'}{a'}} 
    =  r \sum_{a' \in A} v_{\pair{b'}{a'}} 
    =  r w_{b'}
  \end{align*}  
  Hence $Mw = rw$.
  If $w \ne 0$ it follows that $r$ is an eigenvalue of $M$.
  Note that if $v$ is non-negative, then $w \ne 0$. This proves the claim.
\end{proof}

\begin{corollary}\label{corollary:dominant}
  Let $h,\overline{h}$ be morphisms such that $\overline{h} : (\Sigma \times A) \to (\Sigma \times A)^*$
  is an annotation of $h : \Sigma \to \Sigma^*$.
  Then the dominant eigenvalue for $h$ coincides with the dominant eigenvalue for $\overline{h}$.
\end{corollary}

\begin{proof}
  By Lemma~\ref{lem:h:oh} every eigenvalue of $h$ is an eigenvalue of $\overline{h}$.
  Thus the dominant eigenvalue of $\overline{h}$ is greater or equal to that of $h$.
  By Theorem~\ref{thm:perron}, the dominant eigenvalue of a non-negative matrix is a real number $\alpha > 1$
  and its corresponding eigenvector is non-negative.
  By Lemma~\ref{lem:h:oh}, every eigenvalue of $\overline{h}$ with a non-negative eigenvector is also an eigenvalue of $h$. 
  Thus the dominant eigenvalue of $h$ is also greater or equal to that of $\overline{h}$.
  Hence the dominant eigenvalues of $h$ and $\overline{h}$ must be equal.
\end{proof}

\begin{theorem}
  Let $\alpha$ and $\beta$ be  multiplicatively independent real numbers.
 If $v$ is a $\alpha$-substitutive sequence and $w$ is an $\beta$-substitutive sequence,
  then $v$ and $w$ have
 no common non-erasing transducts except for the ultimately periodic sequences. 
\label{theorem-apply-Durand}
\end{theorem}

\begin{proof} 
  Let $h_v$ and $h_w$ be morphisms whose fixed points are $v$ and $w$, respectively.
  By the proof of Theorem~\ref{the:transducts_preserve_morphic}, 
  $x$ is a morphic image of an annotation  $\hbar_v$ 
  of $h_v$, and also of an annotation $\hbar_w$ of $h_w$.  
  The morphisms must be non-erasing, by the assumption in this theorem.
  By  Corollary~\ref{corollary:dominant} and Theorem~\ref{thm:morphism}, 
  $x$ is both $\alpha$- and $\beta$-substitutive.
  By Durand's Theorem~\ref{theorem-Durand}, $x$ is eventually periodic.
\end{proof}

\subsection{Example}

We conclude the section with an example of Theorem~\ref{the:transducts_preserve_morphic}
and the lemmas in this section.

\begin{example}[]
  We saw the Fibonacci sequence  in Example~\ref{example-Fibonacci}:
  $$x = abaababaabaababaababaabaababaabaababaaba\cdots$$
  We conclude our series of examples pertaining to this sequence
  and the transducer $M$
  which doubles every other letter (see Example~\ref{example-transducer} and Figure~\ref{fig:double}).
  We want to exhibit $\hbar$,
  following the recipe of Lemma~\ref{lem:state_annotated_is_morphic}.  First,  some examples of how $\hbar$ works:
  $$\begin{array}{l@{\hspace{.4in}}l}
    (b, \Theta(a)) \mapsto (a, \Theta(ab)) & (a, \Theta(\e)) \mapsto (a, \Theta(\e))(b, \Theta(a)) \\ 
    (b, \Theta(ab))  \mapsto (a, \Theta(aba)) = (a, \Theta(b)) & 
    (a, \Theta(a)) \mapsto (a, \Theta(ab))(b, \Theta(aba)) = (a, \Theta(ab))(b, \Theta(b))
  \end{array}$$

  It turns out that only a few elements from this $A$ end up appearing in the expressions for
  $\hbar(\sigma, \Theta(w))$: 
  It is convenient to abbreviate some of the elements of $\Sigma\times A$:  
  Let us use $x$ as an element of $\set{a,b}$, and  also
  write 
  $(x,\Theta(\e))$ as $x_0$, 
  $(x,\Theta(a))$ as $x_1$,
  $(x,\Theta(b))$ as $x_2$
  and $(x,\Theta(ab))$ as $x_3$.
  It turns out that we do not need to exhibit $\hbar$ in full because only eight points are reachable
  from $a_0$.
  We may take $\hbar$ to be 
  $$\begin{array}{l@{\hspace{.4in}}l@{\hspace{.4in}}l@{\hspace{.4in}}l}
    a_0 \mapsto a_0b_1 &  a_1  \mapsto  a_2 b_3 &  a_2  \mapsto a_3b_2 &  a_3   \mapsto a_1b_0 \\
    b_0  \mapsto a_0 &  b_1  \mapsto  a_2  & b_2  \mapsto a_3 &  b_3  \mapsto a_1
  \end{array}$$

  The fixpoint of this morphism starting with $a_0$ starts as $$ y = \hbar^\omega(a_0) = a_0 \ b_1 \ a_2 \ a_3 \ b_2 \ a_1 \ b_0 \ a_3 \ a_2 \ b_3 \ a_0 \ a_1 \ b_0 \ a_3 \ b_2 \ a_1 \ a_0 \ b_1 \ a_2 \ b_3 \ a_0 \ a_1 \ b_0 \ a_3 \ a_2 \ b_3 \ a_0 \ b_1 \  \cdots$$
  Turning to the coding $c$, recall that the set $Q$ of states of $M$ is $\set{s,t}$. 
  Let us abbreviate the elements of $\Sigma\times Q$ the same way we did
  with $\Sigma\times A$.   
  It is not hard to check that $c(\sigma_0) = \sigma_s$, $c(\sigma_1) = \sigma_t$,   
  $c(\sigma_2) = \sigma_s$, and
  $c(\sigma_3) = \sigma_t$.
  Then the state-annotated
  sequence $z$ from Lemma~\ref{lem:state_annotated_is_morphic} is
  $$z = c(y) = 
  a_s \ b_t \ a_s \ a_t \ b_s \ a_t \ b_s \ a_t \ a_s \ b_t \ a_s \ a_t \ b_s \ a_t \ b_s \ a_t \ a_s \ b_t \ a_s \ b_t \ a_s \ a_t \ b_s \ a_t \ a_s \ b_t \ a_s \ b_t \  \cdots$$

  Recall that $\lambda:\Sigma\times Q \to \Delta^* = \Sigma^*$ 
  in our transducer doubles whatever letter it sees while in state $s$ and copies whatever letter it sees while in state $t$.
  That is, $\lambda(x_s) = xx$, and $\lambda(x_t) = x$.
  Thus when we apply the morphism $\lambda$ to the sequence $z$, we get
  $$\lambda(z) = aa \ b \ aa \ a \ bb \ a \ bb \ a \ aa \ b \ aa \ a \ bb \ a \ bb \ a \ aa \ b \ aa \ b \ aa \ a \ bb \ a \ aa \ b \ aa \ b \  \cdots$$
  As we saw in the proof of Theorem~\ref{the:transducts_preserve_morphic}, this
  sequence 
  $$aa  b  aa  a  bb  a  bb  a  aa  b  aa  a  bb  a  bb  a  aa  b  aa  b  aa  a  bb  a  aa  b  aa  b   \cdots$$
  is exactly $M(x)$.  
\end{example}

\section{Conclusion}

We have re-proven some of the central results in the area of morphic sequences, the closure
of the morphic sequences under morphic images and transduction.    However,
the main results in this paper come from the eigenvalue analyses which followed our proofs
in Sections~\ref{section-morphic-images} and~\ref{section-closure-transduction}.
These are some of the only results known to us which enable one to prove 
negative results on the transducibility relation $\trianglelefteq$.
One such result is in 
Theorem~\ref{theorem-apply-Durand}; this is perhaps the culmination of this paper.

The next step in this line of work is to weaken the hypothesis in some of results that 
the transducers be \emph{non-erasing}. Although our results can be used to reason
about erasing morphisms, see Remark~\ref{rem:erasing:morphism},
this does not help us with erasing transducers since annotating a morphism 
can yield an unbounded large alphabet.
As a consequence, to reason about erasing transducers, we need to understand better
what form of annotated morphisms arise from transducers,
and how these interact with the erasure of letters (Proposition~\ref{prop:close:erase}).

\bibliographystyle{plain}
\bibliography{main}


\end{document}